\documentclass[letterpaper, 10 pt, conference]{ieeeconf}

\IEEEoverridecommandlockouts                              % This command is only needed if you want to use the \thanks command

\usepackage{cite}        % required for bibliography
\usepackage[pdftex]{graphicx}
\usepackage{amsmath,amssymb,epsfig}
\usepackage{enumerate,url,wasysym,balance}
\usepackage{multirow}
\usepackage{xcolor}
\usepackage{algorithm}
\usepackage{algpseudocode}

\usepackage{balance}

\newcommand \mcD{\mathcal{D}}
\newcommand \mcL{\mathcal{L}}
\newcommand \mcN{\mathcal{N}}

\newcommand \mcX{\mathcal{X}}
\newcommand \mcY{\mathcal{Y}}
\newcommand \mcZ{\mathcal{Z}}

\newcommand \E{\mathop{\mathbb{E}}}

\newcommand \1{\mathbf{1}}

\newcommand \pre{\text{pre}}
\newcommand \pref{p_{\text{ref}}}
\newcommand \eq{\text{eq}}

\newcommand \bB{\mathbf{B}}
\newcommand \bI{\mathbf{I}}

\newcommand \bd{\mathbf{d}}
\newcommand \bg{\mathbf{g}}
\newcommand \br{\mathbf{r}}
\newcommand \bs{\mathbf{s}}
\newcommand \bw{\mathbf{w}}
\newcommand \bx{\mathbf{x}}

\newcommand \bz{\mathbf{z}}

\newcommand \bbeta{\boldsymbol{\beta}}
\newcommand \bdelta{\boldsymbol{\delta}}

\newtheorem{proposition}{Proposition}
\newtheorem{remark}{Remark}
\newtheorem{assumption}{Assumption}

\title{\LARGE \bf
Demand Response Under Stochastic, Price-Dependent User Behavior}

\author{Guido Cavraro \and Andrey Bernstein \and Emiliano Dall'Anese
\thanks{G. Cavraro and A. Bernstein are  with the Power Systems Engineering Center, National Renewable Energy Laboratory. email: {\tt guido.cavraro@nrel.gov, andrey.bernstein@nrel.gov}. E. Dall'Anese is  ...}
\thanks{This work was authored by the National Renewable Energy Laboratory, operated by Alliance for Sustainable Energy, LLC, for the U.S. Department of Energy (DOE) under Contract No. DE-AC36-08GO28308. Funding provided by the NREL Laboratory Directed Research and Development Program. The views expressed in the article do not necessarily represent the views of the DOE or the U.S. Government. The U.S. Government retains and the publisher, by accepting the article for publication, acknowledges that the U.S. Government retains a nonexclusive, paid-up, irrevocable, worldwide license to publish or reproduce the published form of this work, or allow others to do so, for U.S. Government purposes.}
}

\begin{document}

\maketitle
\thispagestyle{empty}
\pagestyle{empty}

%%%%%%%%%%%%%%%%%%%%%%%%%%%%%%%%%%%%%%%%%%%%%%%%%%%%%%%%%%%%%%%%%%%%%%%%%%%%%%%%
\begin{abstract}
This paper focuses on price-based residential demand response (DR) implemented through dynamic adjustments of electricity prices during DR events. It extends existing DR models to a stochastic framework in which customer response is represented by price-dependent random variables, leveraging models and tools from the theory of stochastic optimization with decision-dependent distributions. The inherent epistemic uncertainty in the customers' responses renders open-loop, model-based DR strategies impractical. To address this challenge, the paper proposes to
%a twofold approach: (i) use data from prior DR events to estimate the average price elasticity of demand, even if the estimates are coarse or imprecise; and (ii)
employ stochastic, feedback-based pricing strategies to compensate for estimation errors and uncertainty in customer response. The paper then establishes theoretical results demonstrating the stability and near-optimality of the proposed approach and validates its effectiveness through numerical simulations.
\end{abstract}

\section{Introduction}

%{\color{red} ED: Guys, the symbol $r_n$ is confusing, can we just remove it? We can say that $d_n$ is the net power. It is odd to say that $d_n$ can be changed but $r_n$ is fixed.} \\ 

%\ab{I agree, let's remove $r_n$ and fold it into $\hat{d}_n$} \\

%\ab{Should we use ``customers/consumers'' instead of ``users''?}

%{\color{red} ED: Yes, sure. Let's use customers.}

With the growing integration of diverse energy resources across the power grid and the rapid electrification of multiple sectors, maintaining a reliable electricity supply has become increasingly challenging for utilities~\cite{NERC2024}. Residential demand response (DR) has long been a key instrument for system operators (SOs) to manage extreme operating conditions, including peak load events and voltage violations \cite{HAIDER2016166}.

DR strategies vary in both temporal resolution and control mechanisms, and are generally categorized into \emph{direct} and \emph{indirect} load control. In direct load control, SOs assume full authority over consumer loads during specified periods, modifying consumption patterns in real time to address system-level contingencies. In contrast, indirect load control leverages economic incentives---either through dynamic pricing or predefined monetary rewards---to encourage consumers to voluntarily adjust their electricity usage. See, e.g., \cite{SILVA2022100857}, for a comprehensive survey.

Despite their potential, the real-world effectiveness of DR programs remains limited. Direct control schemes often face low participation rates, as they may compromise user comfort or convenience, requiring customers to accept reduced autonomy over their energy use \cite{STENNER201776}. Indirect control strategies, on the other hand, rely heavily on accurate modeling of consumer behavior in response to pricing or incentives. Inadequately designed price signals can inadvertently cause adverse effects, such as load synchronization or system instability, undermining the reliability they are intended to support \cite{Nazir2019}.

This paper focuses on indirect load control, precisely on \emph{real-time price-based} DR. Specifically, the SO dynamically modifies electricity prices during DR events. The classical approach to modeling user response assumes a deterministic relationship between price and consumption, often grounded in the assumption of rational user behavior (e.g.,~\cite{Cavraro2024Incentives,PANDEY2022107597}). However, such assumptions are overly simplistic and fail to capture the inherent variability and uncertainty in human decision-making.

A common extension of deterministic models introduces stochasticity into the user response---typically by adding fixed noise terms~\cite{ReissHousehold2005}. Yet, this approach does not account for the fact that consumer responsiveness is inherently price-dependent: as prices increase, a stronger reduction in demand is generally expected. To address this limitation, we propose a novel framework that models user behavior using price-dependent random variables, drawing on tools from the theory of stochastic optimization with decision-dependent distributions; see the models and algorithms in~\cite{pmlr-v119-perdomo20a,drusvyatskiy2023stochastic,lin2023plug,miller2021outside,wood2021online,Wood2023SaddlePoint}, the more recent works in~\cite{wood2024solving,he2025decision}, and the survey~\cite{hardt2025performative}.

One of the key challenges in implementing indirect, price-based DR strategies lies in the epistemic uncertainty of user response. Specifically, accurately estimating residential price elasticity is difficult because behavioral responses are small, confounded by many factors, and the data (price variation and the corresponding consumption change) are limited~\cite{FELL201437,ReissHousehold2005}. Moreover, factors such as number of household members,  income, and overall financial situation affect the sensitivity to prices~\cite{PEERSMAN2024107421}. Finally, different studies provide different sensitivity results~\cite{ALBERINI2011870}. This inherent epistemic uncertainty makes traditional open-loop, model-based approaches infeasible in practice. Therefore, this paper advocates a 
%two-fold approach: (i) use data from prior DR events to estimate \emph{average} price elasticity of demand (even if inaccurate/rough); and (ii) employ 
stochastic \emph{feedback-based} pricing strategies to compensate for the uncertainty in user response and ensure robust performance.

This paper is structured as follows. In Section \ref{sec:model}, we present the power grid and DR models. In Section \ref{sec:dr}, we formulate the corresponding stochastic optimization problem. In Section \ref{sec:control}, we propose our feedback-based DR strategy and 
establish theoretical results demonstrating its  stability and near-optimality. In Section \ref{sec:num}, we validate its practical effectiveness through numerical simulations. Finally, we conclude in Section \ref{sec:conc}.

%Estimating residential energy price elasticity is difficult because behavioral responses are small~\cite{FELL201437,ReissHousehold2005}, confounded by many factors, and the data (price variation and consumption detail) are limited. 
%Also, prices sensitivities are different among households. Evidence shows that factors like number of house members,  income, overall financial situation affect the sensitivity to prices~\cite{PEERSMAN2024107421}.

%Different studies provide different sensitivity results, e.g., see~\cite{ALBERINI2011870}

%Stochastic models are already employed. Nevertheless, stochasticity is considered non variable~\cite{ReissHousehold2005}

\emph{Notation:} Lower- (upper-) case boldface letters denote
column vectors (matrices).
Given a set $\mcX$ and a vector $\bx$, the projection of $\bx$ onto $\mcX$ is $[\bx]_\mcX$. For a a vector $\bx$, $\|\bx\| := \sqrt{\bx^\top \bx}$. The symbol $\1$ represent the vector of all ones, whereas $\bI$ represent the identity matrix, both of suitable dimension.
$\E_{w \sim \mcD}$ denotes the expected value with respect to random variable $w$ that is distributed according to distribution~$\mcD$.

\section{Model and Demand Response Events} \label{sec:model}

The paper focuses on aggregations of electricity customers participating in DR programs. While the proposed methods are applicable across different geographical scales (e.g., one or multiple distribution circuits), for clarity and concreteness we develop the technical narrative around a specific case involving a feeder. Consider then a distribution feeder with $N+1$ nodes, collected in the set $\mcN = \{0, 1, \dots, N\}$. The substation is labeled as $0$. 
Each of the nodes $n = 1, \ldots, N$ serves a group of electricity users located in close geographical proximity, such as those within a specific neighborhood (and served by the same distribution transformer or connected to the same lateral). The aggregate load at node $n$ is denoted as $d_n$; on the other hand, distributed generation, assume to be uncontrollable, at the same node is denoted as $r_n$. Users are billed for their electricity consumption under a net energy metering (NEM) tariff structure~\cite{Cavraro2024Incentives,Alahmed2024}. Accordingly, we model the total payment collected from all users served at node \( n \) as
\begin{equation}
    \gamma_n(d_n) \;=\; \pi \,(d_n - r_n) \;+\; \omega_n ,
\end{equation}
where \( \pi > 0 \) denotes the retail electricity rate and \( \omega_n \) represents non-volumetric surcharges, such as fixed connection or service fees. Here, we emphasize the dependence of $\gamma_n(d_n)$ on $d_n$, since we will decouple $d_n$ into non-dispatchable  and flexible loads. 

The utility may initiate a \emph{demand response event} (DRE), during which the electricity price is adjusted at each node for users enrolled in the program. Before the DRE, the demand $d_n^{\pre}$ (where the superscript ${\pre}$ indicates that we are considering $d_n$ prior to the DRE) is partitioned into an inflexible demand $\hat d_n$ (e.g., loads from users that will not be adjusted during the DRE, such as refrigerators, etc.) and a pre-event flexible load, whose aggregate consumption is denoted by $\delta_n$ (e.g., dispatchable loads such as washing machines, dryers, EV chargers, or HVAC systems). Overall, before the event,
\[
d_n^{\pre} = \hat d_n + \delta_n - r_n .
\]
Turning now to the DRE, let $x_n$ denote the price adjustment at node $n$ introduced to induce users to lower their power consumption. Accordingly, the cost incurred by a user at node $n$ participating in the DR program is given by
\begin{align*}
    c_n(w_n) 
    &= (\pi + x_n)\big(d_n - r_n\big) + \omega \\
    &= (\pi + x_n)\big(\hat d_n + w_n - r_n\big) + \omega ,
\end{align*}
where $d_n$ has been expressed as
\begin{equation}
\label{eq:power_dem_n}
d_n = \hat d_n + w_n,
\end{equation}
with $w_n \in [0, \delta_n]$ representing the deviation from the baseline $\hat d_n$ implemented by users \emph{in response to the price adjustment}.

A key question pertains to how to model $w_n$. Despite being the classic approach in the literature, modeling the load response to price changes deterministically—for example, by relying on user rationality~\cite{Cavraro2024Incentives}—is an oversimplification. In practice, the responsiveness of users to price signals may be influenced by external factors such as weather conditions, occupancy patterns, appliance availability, or individual behavioral preferences. Therefore, the approach proposed in this paper is to model $w_n$ as a random variable whose distribution depends on the price signal $x_n$. That is, 
$$
w_n \sim \mcD_n(x_n),
$$
where $\mcD_n(x_n)$ is a distribution inducing map,
and $w_n$ is supported on a complete and separable metric space, with support $[0, \delta_n]$. This formulation naturally fits within the framework of stochastic optimization with decision-dependent distributions~\cite{drusvyatskiy2023stochastic,miller2021outside,wood2021online,Wood2023SaddlePoint}, and it effectively models a stochastic, uncertain  response  of the users to a price $x_n$. 
%Overall, during the DRE, the load is given by
%
%\begin{equation}
%\label{eq:power_dem_n}
%d_n = \hat d_n + w_n
%\end{equation}
%with $w_n \sim \mcD_n(x_n)$. 
To ground our model, similarly to~\cite{wood2021online,Wood2023SaddlePoint}, we consider the set of Radon probability measures on a complete and separable metric space $M$ with finite first moment, denoted  as $\mathcal{P}(M)$; when computing expected values, we utilize the probability
measures $m_{(x_n)} \in \mathcal{P}(M)$ that is given as the output of the distributional map $\mcD_n(x_n)$ for each $x_n$; for example, $\E_{w_n \sim \mcD(\bx_n)}[w_n] = \int_M w ~ m_{(x_n)}(d w)$. 

Let
$$\mu_n(x_n) := \E_{w_n \sim \mcD(x_n)}[w_n].$$
denote the mean of $w_n$. In this paper, we assume a \emph{linearized} model for $\mu_n(x_n)$ as follows. 

\vspace{.1cm} 

\begin{assumption}
\label{ass:decr_mu}
For small enough $x_n$, $\mu_n(x_n)$ is a linearly decreasing function of the energy price change:    
\begin{equation}
\label{eq:mu}
\mu_n(x_n) = \delta_n - \beta_n x_n,    
\end{equation}
where $\beta_n > 0$ is the price change sensitivity of bus $n$. \hfill $\Box$
\end{assumption}

\vspace{.1cm} 

Collect all the quantities discussed above in the vectors $\bd^\pre, \hat \bd, \bdelta, \bw  \in \mathbb R^N$. 
A couple of remarks are presented next. 

\vspace{.1cm} 

\begin{remark}
We note that Assumption \ref{ass:decr_mu} is satisfied, for example, if the distributional map $\mcD_n(x_n)$ induces \emph{location-scale family}~\cite{wong2008preferences}:
\begin{align}
     w_n \stackrel{d}{=} \delta_n - \beta_n x_n + \varrho_n
\end{align}
where $\stackrel{d}{=}$ means ``in distribution,'' $\beta_n > 0$ models an (unknown) parameter, and $\varrho_n$ is some stationary zero-mean random variable; see also~\cite{miller2021outside,Wood2023SaddlePoint}. This is a model that captures the sensitivity to small price variations of customers at node $n$ via the distributional shift $\beta_n x_n $, plus an additional element of randomness via $\varrho_n$  to model  external factors. The supports of both $\varrho_n$ and $w_n$ are assumed compact, modeling limited total load flexibility. However, Assumption \ref{ass:decr_mu} is more general, and captures also other (possibly non-linear) distribution models, wherein the mean is approximately linear for small $x_n$.  \hfill $\Box$
\end{remark}

\vspace{.1cm} 

\begin{remark}
Assumption \ref{ass:decr_mu} is justified as $x_n$ is the deviation from the nominal electricity price $\pi$, and it would be small to satisfy price regulation requirements; c.f.~the set $\mcX_n$ in Section \ref{sec:dr}. Moreover, implicitly, we assume that $x_n \leq \delta_n/\beta_n$ so that $\mu_n(x_n) \geq 0$. In general, linear models have been often used in the literature to capture price elasticities, e.g., see~\cite{ALBERINI2011870,MAMKHEZRI2025114537,LI2021120921,FILIPPINI2018137} when load adjustment to price changes is fast. Notably, linear elasticities arise when the interactions between the price makers and energy users is modeled as a Stackelberg game in which user utility of consumption functions are quadratic and the energy prices are linear~\cite{Cavraro2024Incentives,FILIPPINI2018137}.  \hfill $\Box$
\end{remark}

\section{A Demand Response Problem Formulation} \label{sec:dr}

In this section, we formulate a stochastic optimization problem whose solution yields the optimal energy price adjustment to be implemented by the utility. Without loss of generality, we consider the case in which the pre-event power flow through the substation exceeds the desired upper bound. Therefore, during the DRE, the system operator aims to lower the total power demand across the feeder to match a reference $\pref$. 

Let $p_0(\bw)$ denote the total power across the feeder to be controlled by the utility company. This can be, for example, either the power flowing through the distribution substation or the aggregation of the powers at the customers' meters. Either way, we use a model of the form: 
\[
p_0(\bw) = \bs^\top (\hat{\bd} + \bw - \br),
\]
with $\bs \in \mathbb{R}^N$. In particular, if $p_0(\bw)$ represents the power flowing through the distribution substation, then $ \bs^\top$  are the coefficients of a linear approximation of the AC power flow equations~\cite{dall2017optimal}; on the other hand, if one measures the power at the distribution transformers or directly at the customers' meters, then the expression for the total power of interest can be readily written with $\bs = \mathbf{1}$. The goal is to steer $p_0$ below a desired upper bound $\pref$. Equivalently, by introducing the target variable
\begin{equation}
\label{eq:exp_xi}
    \xi(\bw) = p_0(\bw) - \pref ,
\end{equation}
we aim to enforce $\xi(\bw) \leq 0$. Before proceeding, a couple of remarks are in order. 

\vspace{.1cm} 

\begin{remark}
When $p_0(\bw)$ represents the power at the   substation,  we utilize a linearized model  to design the pricing strategy presented in the ensuing section. However, the deployed algorithm will rely on actual measurements of the power at the substation, thereby capturing the nonlinearities of the power flows.  \hfill $\Box$
\end{remark}

\vspace{.1cm} 

\begin{remark}
While we consider DREs with a demand reduction,  the case in which the pre-event power flow through the substation is below a minimum desired value $\pref$ can be treated analogously. This situation may arise, for example, when there is excess generation in the distribution network. \hfill $\Box$
\end{remark}

\vspace{.1cm}

During the DRE, we propose to minimize a cost given by the weighted sum of two components. 
The first represents the cost incurred by the utility during the DRE, and corresponds to the difference between the generation cost and the total payments collected from users, defined as $c_0(\bx,\bw) := \pi_0 p_0(\bw)  - \sum_n c_n(w_n)$. Then, $c_0(\bx,\bw)$ can be further written as: 
\begin{align}
c_0(\bx,\bw)  %= \sum_n (\pi_0 s_n - \pi - x_n) (d_n(w_n) - r_n) - \omega \notag \\
& = \sum_n (\pi_0 s_n - \pi - x_n) (\hat d_n + w_n  - r_n) - \omega \notag \\
& = (\pi_0 \bs - \pi \1 - \bx)^\top (\hat \bd + \bw - \br) - N \omega
\label{eq:c}
\end{align}
where $\pi_0$ is the energy cost for the utility. 

% {\color{red} ED: Note for me to remember symbols: 
% $$\xi(\bw) = p_0(\bw) - \pref$$
% $$d_n(w_n) = \hat d_n + w_n$$
% $$p_0(\bw) = \1^\top (\bd(\bw) - \br)$$
% with $w_n$ a r.v.. 
% Therefore: 
% \begin{align*}
%    p_0(\bw) & = \1^\top (\bd(\bw) - \br) \\
%    & = \underbrace{\sum_n \hat d_n}_{:= \hat d} + \underbrace{\sum_n w_n}_{:=w} - \underbrace{\sum_n r_n}_{:=r} 
% \end{align*}
% }

The second component penalizes the squared norm of $\bx$, reflecting the objective of keeping the energy price as close as possible to the pre-event price $\pi$. In addition, the price adjustment at each node, $x_n$, is constrained to lie within a predefined interval $\mcX_n = [x_{\min}, x_{\max}]$. This constraint protects users from excessive bill fluctuations and captures regulatory limits on price variations. The optimization problem posed to design the  price adjustments $\bx = [x_1, \ldots, x_N^\top]$ results in: 
\begin{subequations}
\label{eq:opt_probl}
\begin{align}
    \min_{\bx \in \mcX} & \; \E_{\bw \sim \mcD(\bx)} [c_0(\bx,\bw)] + \frac \kappa 2 \|\bx\|^2 \label{eq:opt_probl_x}\\
    \text{s.t. }  & \E_{\bw \sim \mcD(\bx)}[\xi(\bw)] \leq 0.
    \label{eq:constr}
\end{align}
\end{subequations}
where $\mcX = \times_{n=1}^{N} \mcX_n$, and where $\bw \sim \mcD(\bx)$ is a short-hand notation for $w_n \sim \mcD_n(\bx_n)$, $\forall \,  n = 1, \ldots, N$.

Using Assumption~\ref{ass:decr_mu}, and noticing that $\E_{\bw \sim \mcD(\bx)}  d(\bw) = \hat \bd + \bdelta - \bB\bx$ and 
\begin{subequations}
\label{eq:expcts}
\begin{align}
& \E_{\bw \sim \mcD(\bx)}  \xi(\bw) = - \bs^\top \bB \bx + \tilde p_0 \label{eq:Exi}\\
& \E_{\bw \sim \mcD(\bx)}  c_0(\bx,\bw) = \bx^\top \bB \bx + \bx^\top \big(\bB ( \pi \1 - \pi_0 \bs) - \bd^\pre \big) + \notag \\
& \qquad \qquad + (\pi_0 \bs^\top - \pi \1^\top)  \bd^\pre - N \omega \label{eq:c}  
\end{align}    
\end{subequations}
where $\bbeta = [\beta_1, \ldots, \beta_N]^\top$ and $\bB = \text{diag}(\bbeta) \in \mathbb R^{N \times N}$, and with $\tilde p_0 = \bs^\top \bd^\pre - \pref$ for brevity, problem~\eqref{eq:opt_probl} can be re-written as the following convex linearly constrained quadratic program (LCQC):
\begin{subequations}
\label{eq:LCQP}
\begin{align}
    \min_{\bx \in \mcX} & \; \bx^\top \Big(\bB + \frac \kappa 2 \bI \Big) \bx + \bx^\top \big(\bB ( \pi \1 - \pi_0 \bs) - \bd^\pre \big)\\
\text{s.t. } & \quad - \bs^\top \bB \bx + \tilde p_0 \leq 0.
\end{align}
\end{subequations}
Unfortunately, solving~\eqref{eq:LCQP} is often impractical -- if not infeasible -- for the system operator for two main reasons: (i) the price sensitivities $\bbeta$ are unknown and difficult to estimate accurately~\cite{lin2023plug,bracale2024learning,wood2024solving}; and (ii) instead of having full knowledge of all loads (both flexible and inflexible) at the customer side in real time, the system operator typically only has access to real-time measurements of $p_0(\bw(t))$. These challenges are particularly critical for enforcing the constraint~\eqref{eq:constr}, since inaccurate estimates of $\bbeta$ may lead to violations, and achieving real-time full load observability would require pervasive metering infrastructure.

\section{Online Feedback Optimization Algorithm} \label{sec:control}

In this section, we introduce a feedback-based stochastic algorithm to steer the price toward an approximated solution of~\eqref{eq:opt_probl}. Our proposed approach combines tools from performative optimization~\cite{hardt2025performative,Wood2023SaddlePoint}, regularized Lagrangian methods~\cite{koshal2011multiuser}, and feedback-based primal-dual algorithms~\cite{dall2016optimal}.

\subsection{Online stochastic algorithm} 
To begin, the Lagrangian function associated with~\eqref{eq:opt_probl} is: 
\begin{equation*}
\mcL(\bx,\lambda) = \E_{\bw \sim \mcD(\bx)}\Big[ c_0(\bx,\bw) \Big] + \frac \kappa 2 \|\bx\|^2 + \lambda \E_{\bw \sim \mcD(\bx)} \Big[ \xi(\bw)\Big]    
\end{equation*}
with $\lambda \geq 0$ the dual variable associated with~\eqref{eq:constr}. Under the current setup and assumptions, $\mcL(\bx,\lambda)$ is strongly convex in $\bx$ for any $\lambda \geq 0$, and concave in $\lambda$ for any $\bx \in \mcX$ (in fact, linear). A common approach to ensure linear convergence of  primal-dual methods is to leverage a regularized Lagrangian~\cite{koshal2011multiuser}: 
\begin{equation}
\label{eq:reg_Lagr}
\hat \mcL(\bx,\lambda) = \mcL(\bx,\lambda) - \frac \eta 2 \lambda^2
\end{equation}
with $\eta > 0$ a regularization coefficient. Hereafter, the \emph{unique} saddle point of~\eqref{eq:reg_Lagr} is referred to as \emph{optimal point} of the regularized Lagrangian, is denoted as $\bz^\star := [(\bx^\star)^\top, \lambda^\star ]^\top$, and it is defined as
$$\bx^\star = \arg \min_{\bx \in \mcX} \max_{\lambda \in \mcY} \hat\mcL(\bx,\lambda), \quad \lambda^\star = \arg \max_{\lambda \in \mcY} \min_{\bx \in \mcX} \hat \mcL(\bx,\lambda).$$
where $\mcY \subset \mathbb{R}_{>0}$ is a finite interval that contains $\lambda^\star$ (built as discussed in, e.g.,~\cite{koshal2011multiuser}).
It is known that $\bz^\star$ deviates from the saddle points of the  Lagrangian function $\mcL(\bx,\lambda)$ when the constraint is active at optimality; fortunately, a bound on the distance between $\bz^\star$ and  saddle points of $\mcL(\bx,\lambda)$ as a function of $\eta$ is provided in~\cite[Prop.~3.1]{koshal2011multiuser}. A standard primal-dual methods applied to~\eqref{eq:reg_Lagr} then amounts to the following iterative steps (where $t \in \mathbb{N}$ is the iteration index):   
\begin{subequations}
\label{eq:PD_PO_gen}
\begin{align}
\bx(t+1) & = \big[\bx(t) - \epsilon \nabla_\bx \hat \mcL(\bx(t),\lambda(t))\big]_\mcX 
 \label{eq:PD_PO_gen_x} \\% \tag{POa}\\
\lambda(t+1) & = \big[(1 + \epsilon \eta) \lambda(t) + \epsilon \, \mathbb{E}_{\bw \sim \mcD(\bx)} [ \xi(\bw)] \big]_\mcY \label{eq:PD_PO_gen_l} % \tag{POb}
\end{align}
\end{subequations}
where $\epsilon>0$ is a step size and
\begin{align}
\hspace{-.2cm} \nabla_\bx \hat \mcL(\bx,\lambda) \!=\!2\Big(\bB + \frac \kappa 2 \bI\Big) \bx + \bB(\pi \1 - \pi_0 \bs) - \bd^\pre - \lambda \bbeta. \hspace{-.1cm}
\label{eq:primal_grad}
\end{align}
We will refer to~\eqref{eq:PD_PO_gen} as  \emph{Performative Optimum} (PO)
algorithm, borrowing the terminology from~\cite{lin2023plug,Wood2023SaddlePoint}. As discussed in Section~\ref{sec:model}, it is infeasible to perform the updates~\eqref{eq:PD_PO_gen} in practice without pervasive metering. The proposed strategy is then to modify~\eqref{eq:PD_PO_gen} by: (i) leveraging estimates 
$\hat \bbeta$ of $\bbeta$~\cite{lin2023plug,bracale2024learning,wood2024solving} in~\eqref{eq:PD_PO_gen_x}, and leveraging a single-measurement stochastic approximation of $\E_{\bw \sim \mcD(\bx)} [ \xi(\bw)] $ in~~\eqref{eq:PD_PO_gen_l}. Accordingly, the proposed online algorithm amounts to the following iterations: 
\begin{subequations}
\label{eq:PD_In_st}
\begin{align}
&\bx(t+1) = \big[\bx(t) - \epsilon \, \bg(\bx(t),\lambda(t))\big]_\mcX \label{eq:PD_In_st_x} \\
&\lambda(t+1) = \big[(1+ \epsilon \eta)\lambda(t) + \epsilon(p_0(t) - \pref) 
\big]_\mcY \label{eq:PD_In_st_l}
\end{align}
\end{subequations}
where $p_0(t)$ is a \emph{measurement} of $p_0(\bw(t))$ at time $t$ (and, thus, $p_0(t) - \pref$ is a measurement of the regulation error), and where 
\begin{equation}
\label{eq:g}
\bg(\bx,\lambda) := 
2\Big(\hat \bB + \frac \kappa 2 \bI\Big) \bx + \hat \bB(\pi \1 - \pi_0 \bs) - \bd^\pre - \lambda \hat \bbeta \,
\end{equation}
where $\hat \bB = \text{diag}(\hat \bbeta)$.
We will refer to~\eqref{eq:PD_In_st} as the \emph{Stochastic Inexact Performative Optimum} (Stochastic InPO) algorithm, which is summarized as~Algorithm~\ref{alg:inpo} and illustrated in Figure~\ref{fig:F_diagram}.

\begin{figure}[t!]
\centering
\includegraphics[width=0.8\columnwidth]{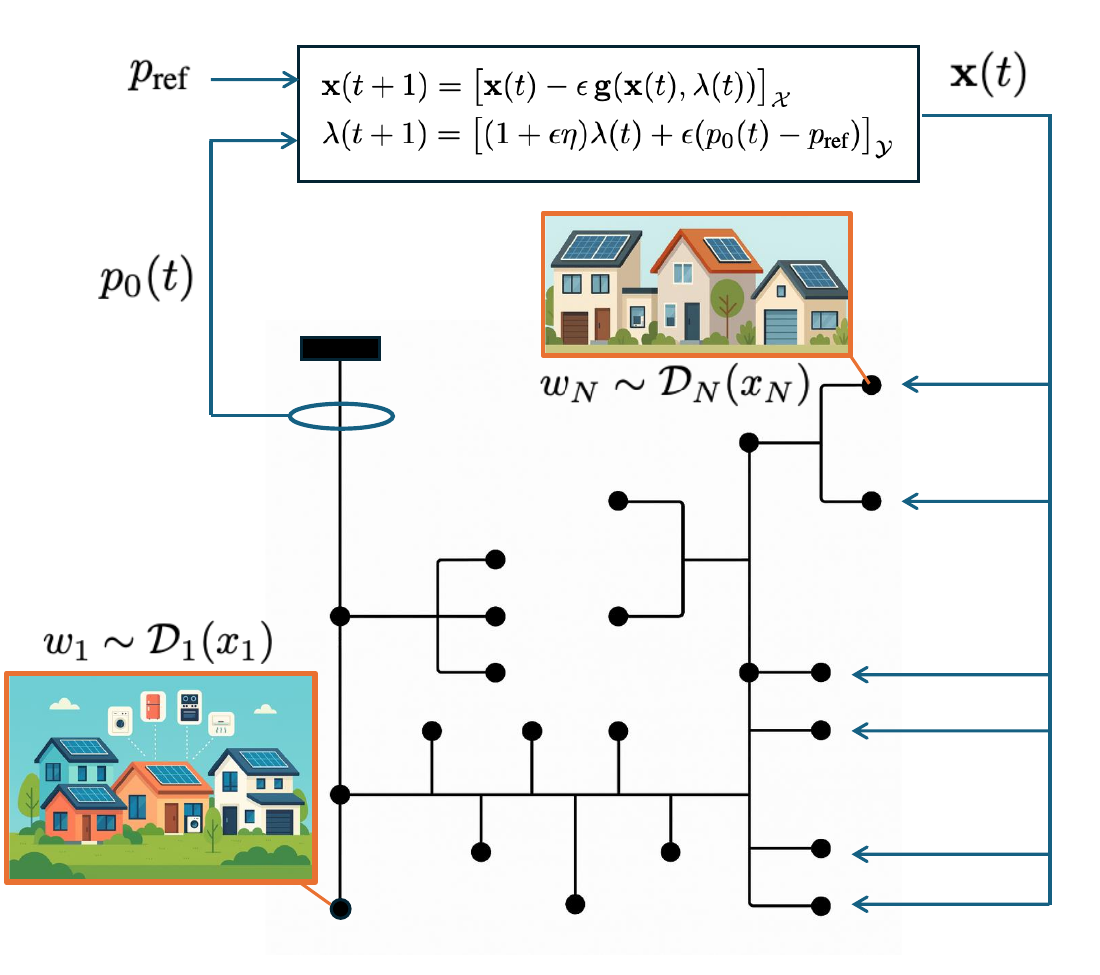}
\caption{Diagram illustrating the steps of Algorithm~\ref{alg:inpo}. This is a scenario where $p_0(t)$ represents a measurement of the power at the point of common coupling.}
\label{fig:F_diagram}
\end{figure}

%%%%%%%%
\begin{algorithm}[t!]
\caption{\emph{Stochastic InPO algorithm}}
\label{alg:inpo}
\textbf{Data gathering}: Compute $\bs$. Estimate $\hat \bbeta$. 

\textbf{Initialization}:  Set $\epsilon > 0$, $\eta > 0$. Set $\bx(0)$, $\lambda(0)$.

\textbf{Real-time operation}: for $t \geq 1$, repeat:

\quad [S1] Update $\bx(t)$ via~\eqref{eq:PD_In_st_x}. 

\quad [S2] Dispatch price adjustments to customers. 

\quad [S3] Measure $p_0(t)$.

\quad [S4] Update $\lambda(t)$ via~\eqref{eq:PD_In_st_l}.

\quad [S5] Go to [S1].
\end{algorithm}
%%%%%%%%

One may utilize a stochastic approximation, based on realizations of $\bw(t)$, also for~\eqref{eq:PD_In_st_x}; however, this would require real-time metering capabilities from each node of the network; instead, we leverage estimates 
$\hat \bbeta$ to bypass this practical challenge. It is also worth pointing out that, if an estimate of $\bbeta$ is not available,  one can set $\hat \beta_n = 1$ for all $n = 1, \ldots, N$; this represents a ``sign approximation''.   

In the following sections, we provide a discussion to compare Algorithm~\ref{alg:inpo} with alternative approaches, and then we provide results for the convergence to $\bz^\star$.

\subsection{Discussion} 

In this section, we discuss limitations of alternative approaches for the design of pricing strategies.

Consider again the PO algorithm~\eqref{eq:PD_PO_gen}, which is a standard primal-dual strategy for identifying the saddle point of~\eqref{eq:reg_Lagr}. By using the expectations~\eqref{eq:expcts}, the PO is rewritten as: 
\begin{subequations}
\label{eq:PD_PO}
\begin{align}
\bx(t+1) &= \big[\big((1 - \epsilon \kappa)\bI - 2 \epsilon \bB\big)\bx(t) + \epsilon \bbeta \lambda(t)+ \notag \\
& \qquad  - \epsilon \bB(\pi \1 - \pi_0 \bs) + \epsilon \bd^\pre) \big]_\mcX 
\\% \tag{POa}\\
\lambda(t+1) & = \Big[(1 - \epsilon \eta)\lambda(t) - \epsilon \bs^\top \bB \bx(t) + \epsilon \tilde p_0 \Big]_\mcY % \tag{POb}
\end{align}
\end{subequations}
where $t$ is again the iteration index. If the coefficients $\bbeta$ are available, the system operator could run~\eqref{eq:PD_PO} offline to convergence without the need of feedback from the system to calculate the prices. However, the coefficients $\bbeta$ are not usually known. Therefore, instead of$\bbeta$, the system operator may rely on an estimate $\hat \bbeta$ and use it to run the following inexact version of~\eqref{eq:PD_PO}:
\begin{subequations}
\label{eq:InPO}
\begin{align}
\bx(t+1) & = \Big[ \big( (1 - \epsilon \kappa) \bI -2\epsilon \hat \bB \big) \bx(t) + \notag \\
& \qquad + \epsilon \hat \bbeta \lambda(t) - \epsilon(\pi - \pi_0) \hat \bbeta + \epsilon \bd^\pre) \big]_\mcX 
\\
\lambda(t+1) &= \Big[(1 - \epsilon \eta)\lambda(t) - \epsilon \bs^\top \bB \bx(t) + \epsilon \tilde p_0 \Big]_\mcY.
 \end{align}
\end{subequations}
We will refer to~\eqref{eq:InPO} as the \emph{Inexact Performative Optimum} (InPO) algorithm. Algorithm~\eqref{eq:InPO} can be run to convergence offline, too. The main limitation of~\eqref{eq:InPO} is that its performance depends on the precision of the estimated sensitivities (and we will elaborate on this point in the experiments in Section~\ref{sec:num}).

A common approach when dealing with decision dependent probability distributions is to seek \emph{stable points} $(\bx_\eq,\lambda_\eq)$ rather than optimal points~\cite{Wood2023SaddlePoint,pmlr-v119-perdomo20a}. Stable points (also termed equilibrium points in some prior works) are optimal with respect to the distribution they induce; that is~\cite{Wood2023SaddlePoint}:
\begin{align*}
&\bx_\eq = \arg \min_{\bx \in \mcX} \max_{\lambda \in \mcY} \E_{\bw \sim \mcD(\bx_\eq)} \phi(\bx,\lambda,\bw)\\
&\lambda_\eq = \arg \max_{\lambda \in \mcY} \min_{\bx \in \mcX} \E_{\bw \sim \mcD(\bx_\eq)} \phi(\bx,\lambda,\bw)
\end{align*}    
with
$$ \phi(\bx,\lambda,\bw) =  c_0(\bx,\bw) + \frac \kappa 2 \|\bx\|^2 + \lambda  \xi(\bw)$$
The following \emph{Performative Stable} (PS) algorithm, developed based on~~\cite{Wood2023SaddlePoint}, can be used to converge to equilibrium points:
\begin{subequations}
\label{eq:PerfSt}
\begin{align}
&\bx(t+1) = \Big[\big((1-\epsilon \kappa )\bI - \epsilon \bB\big)\bx(t) + \epsilon \bd^\pre  \Big]_\mcX \label{eq:PerfSt_x}\\
&\lambda(t+1) = \Big[(1 - \epsilon \eta)\lambda(t) - \epsilon \bs^\top \bB \bx(t) + \epsilon \tilde p_0 \Big]_\mcY.
\end{align}    
\end{subequations}
Algorithm~\eqref{eq:PerfSt} presents two main limitation. The first is that it relies on the knowledge of the sensitivities $\bbeta$. But, even of one uses estimates of $\bbeta$, the price update is decoupled from the one of the Lagrange multiplier $\lambda$. That is, the price does not take into account violations of the operational constraint~\eqref{eq:constr}. 

The three algorithms presented in this subsections will be used as a comparison with the proposed  Stochastic InPO in Section~\ref{sec:num}.

\subsection{Convergence} 

In this section, we analyze the convergence of the online stochastic algorithm~\eqref{eq:PD_In_st}. We start with the static setup, where the non-dispatchable loads are fixed during the DRE.  

To state the main convergence result, we set $\bz := [\bx^\top, \lambda, ]^\top$ and define  $\mcZ := \mcX \times \mcY$. Also, define the following map: 
\begin{align}
    \label{eq:truemap}
    G(\bz) := 
    \begin{bmatrix}
        \nabla_\bx \hat \mcL(\bx(t),\lambda(t)) \\
        - \nabla_\lambda \hat \mcL(\bx(t),\lambda(t))
    \end{bmatrix}
\end{align}
Using~\eqref{eq:truemap}, the algorithm~\eqref{eq:PD_PO} for finding the PO can be re-written as $\bz(t+1) = \left[\bz(t) - \epsilon G(\bz(t)) \right]_\mcZ$.  The map $\bz \mapsto G(\bz)$ is $\nu$-strongly monotone and $L$-Lipchitz over $\mcZ$, where $\nu = \min\{\kappa + 2 \min{\beta_n},\eta\}$ and where $L$ can be computed as in~\cite[Lemma~3.4.]{koshal2011multiuser}. Note that $\bz^\star$, by definition, satisfies the fixed-point equation $\bz^\star = \text{proj}_{\mcZ} \left\{\bz^\star - \epsilon G(\bz^\star) \right\}$. Regarding the stochastic iteration~\eqref{eq:PD_In_st}, we have the following assumption. 

\vspace{.1cm}

\begin{assumption}
\label{as:boundederror}
    There exists $\bar{e}_\xi < + \infty$ such that 
    $$
    \E_{\bw \sim \mcD(\bx(t))} \Big[   \Big|\xi(\bw(t)) - \E_{\bw \sim \mcD(\bx(t))} [\xi(\bw)]  \Big|  \Big]\leq \bar{e}_\xi 
    $$
    for any $\bx \in \mcX$. \hfill $\Box$
\end{assumption}

\vspace{.1cm}

The assumption implies that the error in incurred by approximating the expected value of $\xi(\bw)$ (which is computable if the distribution is known and all the loads can be measured) with the single-observation estimate $\xi(\bw(t))$ is bounded. Then, we have the following result. 

\vspace{.1cm}

\begin{proposition}[Convergence with static nonflexible loads]
\label{prop:static}
    Consider the stochastic algorithm~\eqref{eq:PD_In_st}, and let $\bz(t)$, $t \in \mathbb{N}$, be the sequence of prices and dual variables generated by~\eqref{eq:PD_In_st}, starting from $\bz(0) \in \mcZ$. Let Assumption~\ref{as:boundederror} hold. Let $\epsilon$ be such that  $\epsilon < 2\nu/L^2$. 
    Then, the error $e(t) := \|\bz(t) - \bz^\star\|$ can be bounded as
    \begin{align}
    \E[e(t)] \leq & ~ c(\epsilon)^{t} \E[e(0)] +  \frac{\epsilon (\|\hat{\bbeta} - \bbeta\| B  + \bar{e}_\xi)}{1 - c(\epsilon)} ,  \,\, t \in \mathbb{N}
\end{align}
where  $c(\epsilon) := (1- 2 \epsilon \nu + \epsilon^2 L^2)^\frac{1}{2} < 1$ and $B := 2 X + |\pi - \pi_0| + \Lambda $, with $X = \max_{\bx \in \mcX} \|\bx\|_2$ and $\Lambda = \max_{\lambda \in \mcY} |\lambda|$. \hfill $\Box$
\end{proposition}

\vspace{.1cm}

The proof Proposition~\ref{prop:static} can be found in the appendix. As expected, $e(t)$ exhibits, on average, a decaying trend up to an error floor that depends on the mismatch between the true sensitivities $\bbeta$ and their estimates, as well as on the stochastic error arising from approximating the expectation of the power $p_0$ with a single measurement.

We now turn the attention to a dynamic case, where $\hat d_n(t)$ may now be time-varying to reflect the case where customers switch on and off non-dispatchable loads during the DRE event. In this case,~\eqref{eq:opt_probl} is a parametric problem with a time-varying   parameter $\hat d_n(t)$; accordingly,~\eqref{eq:reg_Lagr} is time-varying and $\bz^\star(t)$ is now a trajectory. Similarly to~\cite{dall2016optimal}, we have the following tracking result proved in the appendix.  

\vspace{.1cm}

\begin{proposition}[Tracking of prices with dynamic nonflexible loads]
\label{prop:dynamic}
    Consider the stochastic algorithm~\eqref{eq:PD_In_st}, and let $\bz(t)$, $t \in \mathbb{N}$, be the sequence of prices and dual variables generated by~\eqref{eq:PD_In_st}, starting from $\bz(0) \in \mcZ$. Let $\bz(t)^\star$, $t \in \mathbb{N}$ be the sequence of optimal points.  Let Assumption~\ref{as:boundederror} hold. Let $\epsilon$ be such that  $\epsilon < 2\nu/L^2$. 
    Then, the error $e(t) := \|\bz(t) - \bz^\star(t)\|$ can be bounded as
    \begin{align}
    \E[e(t)] \leq & ~ c(\epsilon)^{t} \E[e(0)] +  \frac{\Delta + \epsilon (\|\hat{\bbeta} - \bbeta\|_2 B  + \bar{e}_\xi)}{1 - c(\epsilon)}
\end{align}
for any $ t \in \mathbb{N}$, where $\Delta := \sup_{t \geq 1} \|\bz^\star(t) - \bz^\star(t-1)\|$. \hfill $\Box$
\end{proposition}

\vspace{.1cm}

The tracking error in this case depends on the maximum variability of the non-controllable loads via the term $\Delta$. 

%
% We stress the fact that perfect knowledge of the probability distributions or, under Assumption~\ref{ass:decr_mu}, of the price sensitivities $\bbeta$, is needed to directly solve~\eqref{eq:LCQP}.
% However, this is often unrealistic as price sensitivities are quickly changing.. \ab{this should be motivated better. I don't think you solve this problem here as you claim you use an estimate of $B$, so is it feasible to estimate it?}

\section{Numerical Examples} \label{sec:num}

We conducted case studies on the IEEE 37-bus feeder~\cite{kersting2018distribution}. We omitted the voltage regulators and converted the network to its single-phase equivalent, as shown in Fig.~\ref{fig:ieee37}. The feeder includes 25 buses with non-zero loads.

\begin{figure}[tb]
\centering
\includegraphics[width=0.45\columnwidth]{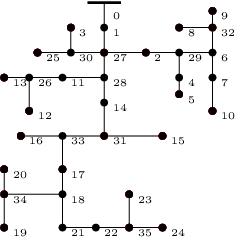}
\caption{IEEE 37-bus feeder. Bus 0 represents the grid substation.}
\label{fig:ieee37}
\vspace{-4mm}
\end{figure}

We tested the proposed algorithms in the two scenarios described below. In the first, we numerically validate the convergence properties of the proposed algorithm by considering constant non-flexible loads. In the second, we evaluate the algorithm’s performance under a realistic condition with time-varying non-flexible loads.

The stochastic user responses are modeled as follows.
For each user $n$, we assume that $w_n$ is drawn from a Gaussian distribution, $w_n \sim \mcN(\mu_n(x_n), \sigma_n(x_n))$. Here, $\mu_n(x_n)$ has the same form in~\eqref{eq:mu}, and $\sigma_n(x_n)$ is chosen as a positive, decreasing function of $x_n$, given by
$$\sigma_n(x_n) = \sigma_n^0 e^{- \zeta \beta_n x_n}.$$
where $\sigma_n^0$ and $\zeta$ are positive constants.
Since the literature reports a broad range of sensitivity values, the price sensitivities ${\beta_n}$ were drawn randomly from a uniform distribution over $[0,2]$~\cite{ReissHousehold2005}.
The DREs we considered featured desired setpoints $\pref$ that can be achieved without curtailing non-flexible demands, i.e.,
$\pref \in [\bs^\top \hat \bd, \bs^\top (\hat \bd + \bdelta)]$.
We set $\kappa = 5$ and $\eta = 0.01$, and used normalized prices $\pi_0 = 1$ and $\pi = 2$.
Notice that the PO, PS, and InPO algorithms can be run to convergence without requiring feedback from the grid. Hence, once the DRE begins, the system operator can simply dispatch the prices obtained after convergence of~\eqref{eq:PD_PO}, \eqref{eq:PerfSt}, and~\eqref{eq:InPO}.

\subsection{Constant Load Scenario}

For each bus $n$, the non-flexible load $\hat d_n$ was set equal to the nominal testbed load, whereas the flexible load $\delta_n$ was randomly chosen from a uniform distribution over $[0, 2 \hat d_n]$.

We first consider the case where $\bbeta$ is estimated as $\hat \bbeta = \1$, i.e., adopting the ``sign approximation''.
We compared the performance of the Stochastic InPO with those of the PO, PS, and InPO algorithms, with results reported in Figures~\ref{fig:P0_MC} and~\ref{fig:Cost_MC}.
The trajectories for the Stochastic InPO were obtained via 300 Monte Carlo simulations.

Figure~\ref{fig:P0_MC} shows the power exchanged with the external network. Both the PS and InPO yield average exchanged power levels far from $\pref$: as mentioned earlier, the PS price update does not account for constraint violations (see~\eqref{eq:PerfSt_x}), while the InPO is affected by model inaccuracies. Conversely, both the PO and the Stochastic InPO successfully regulate $p_0$ close to $\pref$: the PO does so by leveraging full knowledge of the probability distributions, and the Stochastic InPO by exploiting feedback from the system.

Figure~\ref{fig:Cost_MC} displays the operational cost achieved by each strategy. The PS and InPO attain the lowest costs, which is not surprising since they both fail to meet the power constraint. The PO and Stochastic InPO instead produce prices close to the optimal ones.
The black line represents the optimum of problem~\eqref{eq:LCQP}.
The small gap between $\pref$ and the exchanged power under the PO and Stochastic InPO arises because they are designed to converge to the saddle point of $\hat \mcL(\bx,\lambda)$ rather than that of $\mcL(\bx,\lambda)$.

Finally, we tested the Stochastic InPO under different values of $\hat \bbeta$.
Figure~\ref{fig:Dist_betas} shows the distance between the average price $\mathbb E[\bx(t)]$ and the optimal price, evaluated as $\|\bx(t) - \bx^*\|$, for different estimates $\hat \bbeta$. Performance deteriorates with increasing estimation error, as expected from Proposition~\ref{prop:static}.
Nevertheless, Figure~\ref{fig:P0_betas} shows that, regardless of the chosen $\hat \bbeta$, the power constraint is satisfied. This demonstrates the benefit of a feedback-based approach: adopting a rough estimate of the sensitivities affects efficiency but does not prevent meeting the operational constraints.

\begin{figure}[tb]
\centering
\includegraphics[width=0.9\columnwidth]{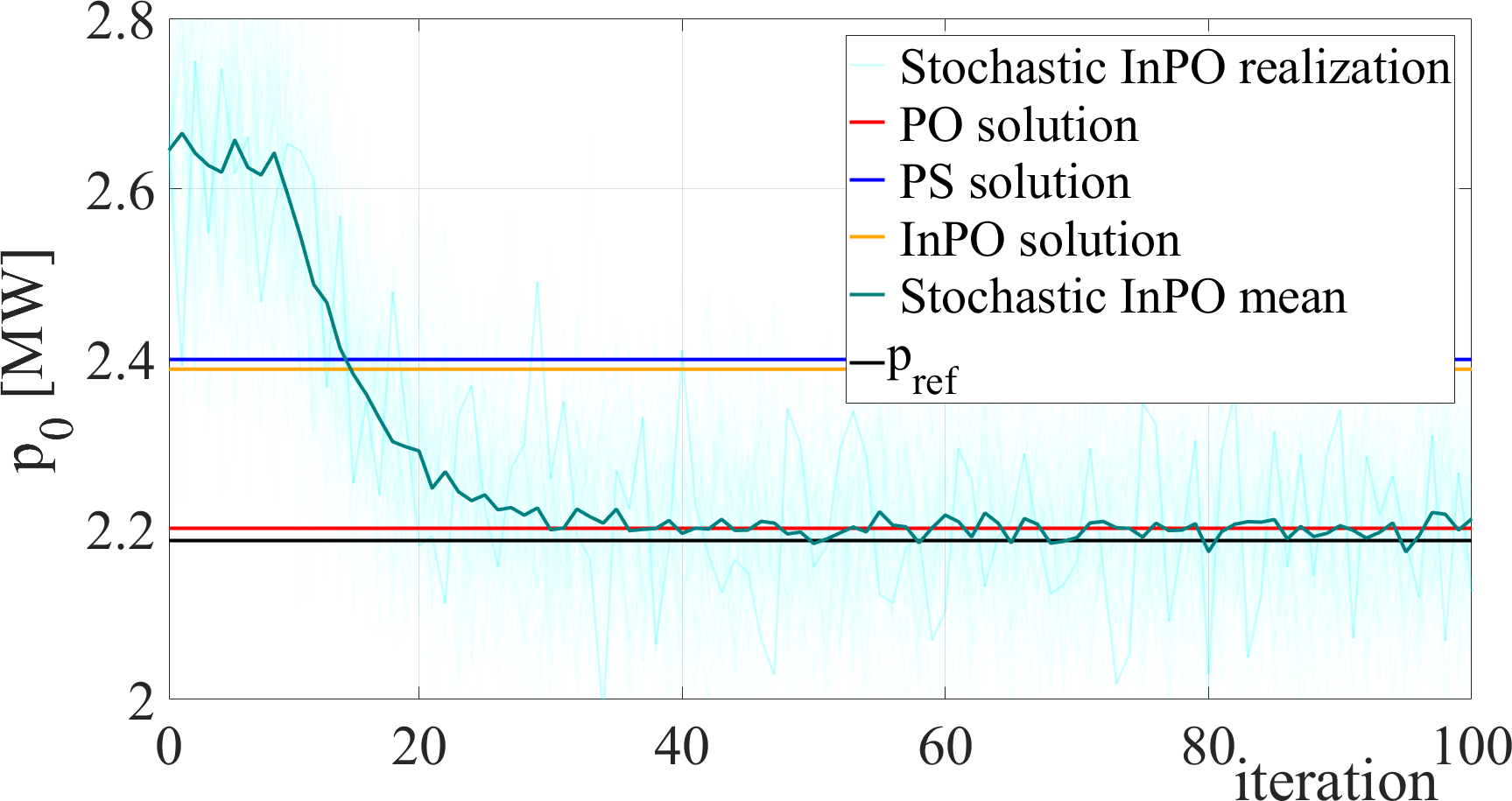}
\caption{Average power exchanged with the external network. %The optimal solution coincides with $\pref$.
}
\label{fig:P0_MC}
\vspace{-5mm}
\end{figure}

\begin{figure}[tb]
\centering
\includegraphics[width=0.9\columnwidth]{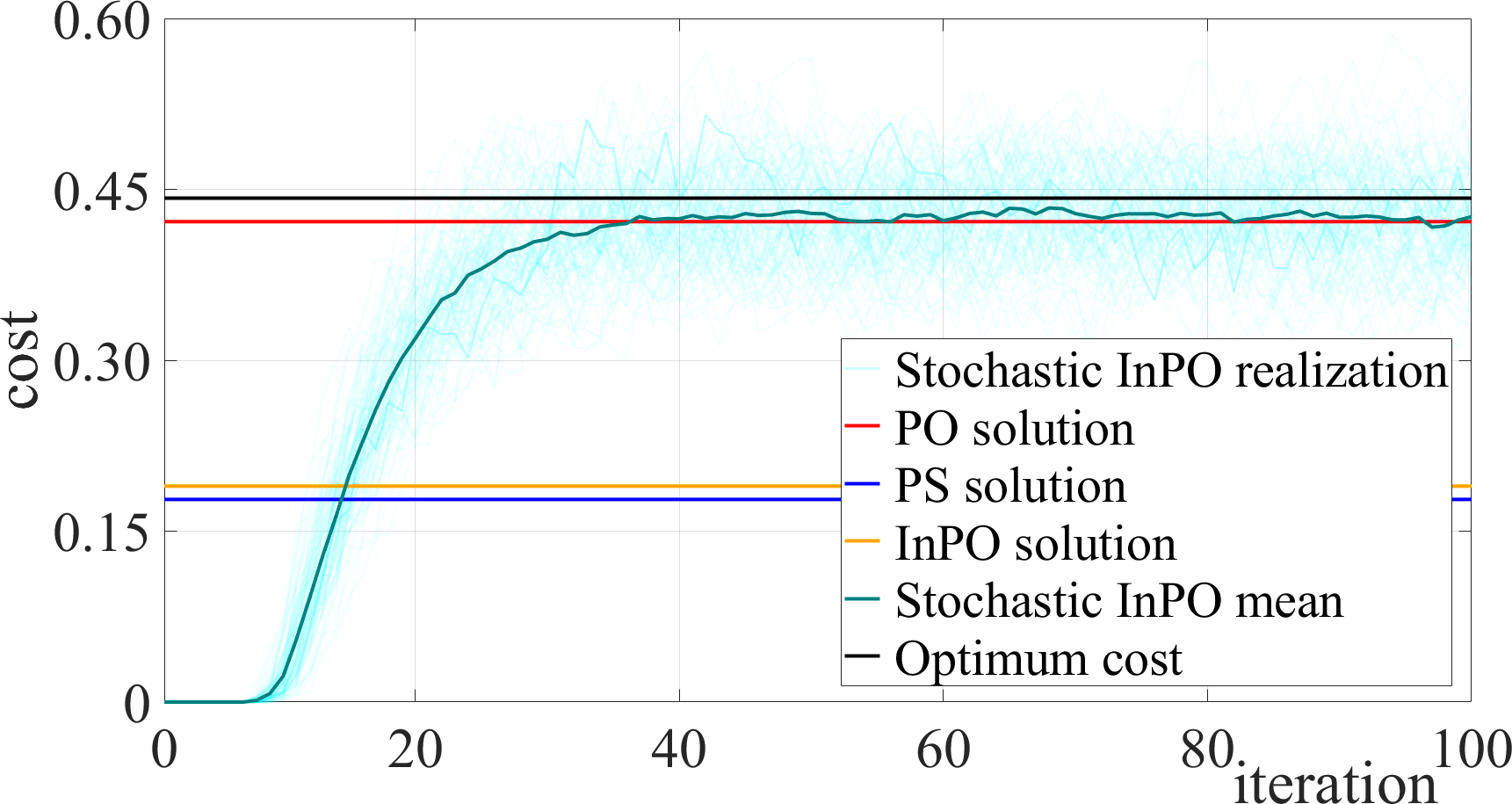}
\caption{Average cost of operating the power grid during the DRE, computed according to~\eqref{eq:opt_probl_x}. %The PS and InPO achieve lower costs but fail to meet the desired load reduction level.
}
\label{fig:Cost_MC}
\end{figure}

\begin{figure}[tb]
\centering
\includegraphics[width=0.9\columnwidth]{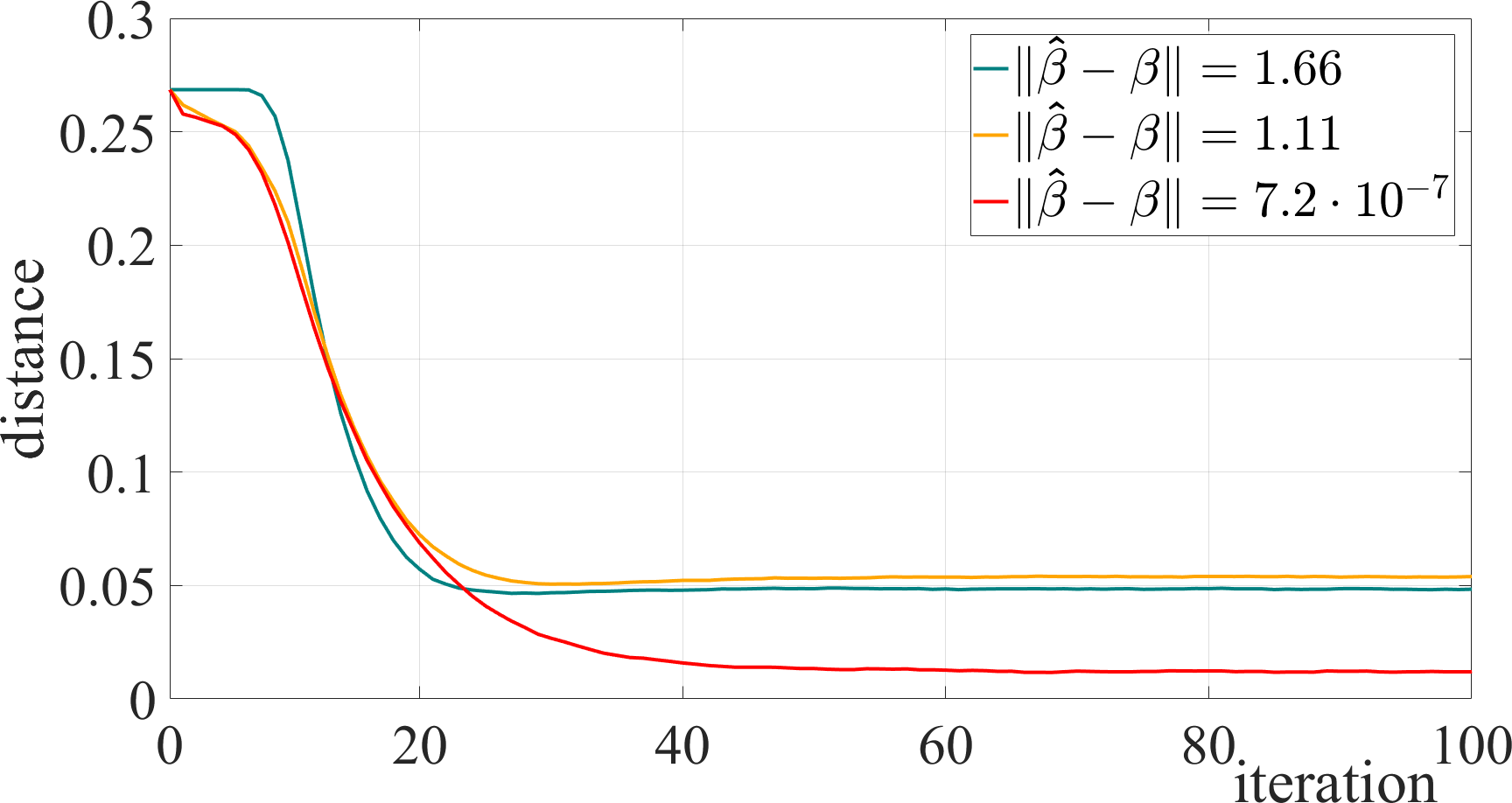}
\caption{Distance between the price vector $\bx(t)$ and the optimal prices $\bx^*$. The distance increases with larger estimation errors affecting $\hat \bbeta$.}
\label{fig:Dist_betas}
\end{figure}

\begin{figure}[tb]
\centering
\includegraphics[width=0.9\columnwidth]{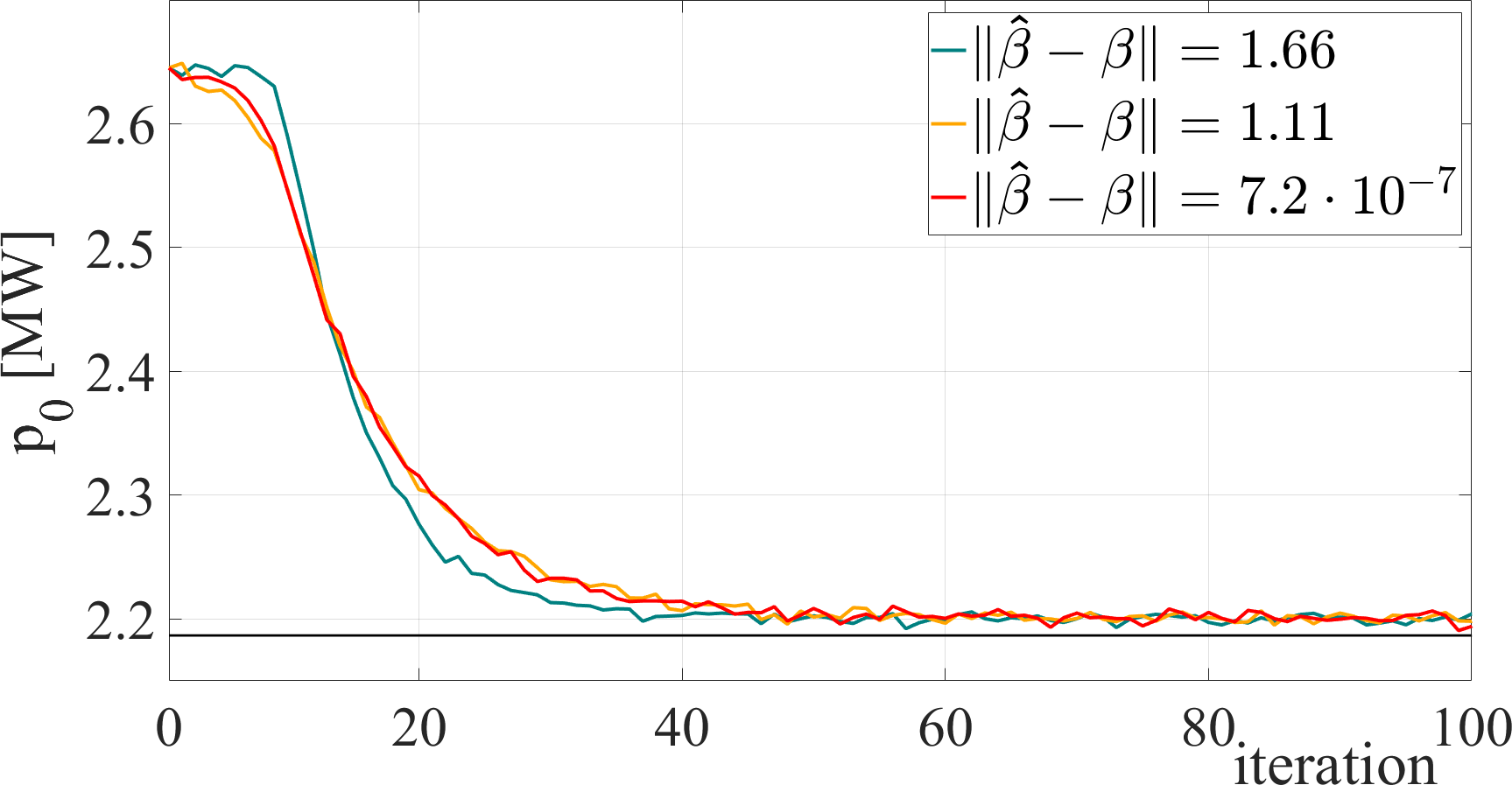}
\caption{Average power exchanged with the external network. 
%The black line represents $\pref$. %The Stochastic InPO satisfies the desired load reduction for different choices of $\hat \bbeta$.
}
\label{fig:P0_betas}
\end{figure}

\subsection{Time-varying Load Scenario}

Each bus hosted 25 residential users. Minute-level data collected from the Pecan Street project on January 1, 2013, were used as power profiles for houses with random assignments~\cite{PecanData}.
These net power profiles correspond to residences both with and without solar generation.
After aggregation, the profiles were normalized so that the peak demand at each bus matched the nominal testbed loads.

We considered a DRE lasting from 10~am to 3~pm.
Figures~\ref{fig:PO_MC_TVa} and~\ref{fig:Cost_MC_TV} show the average power exchange and operational cost over 300 Monte Carlo simulations.

Again, the operational cost resulting from the PS is the lowest, as it fails to regulate $p_0$. Indeed, the PS price updates do not account for constraint violations.
Unlike in the static case, here the InPO over-satisfies the power constraint, yielding $p_0 < \pref$. This over-satisfaction leads to the highest operational cost among all strategies.
By leveraging perfect knowledge of the probability distributions, the PO achieves the best overall performance, resulting in the lowest cost among all algorithms that steer $p_0$ close to $\pref$.
The Stochastic InPO is slightly less efficient in terms of operational cost but still regulates $p_0$ satisfactorily.

The numerical simulations show the effectiveness and robustness of the proposed real-time pricing strategy. The Stochastic InPO consistently meets the exchanged power constraints under uncertain user sensitivities and varying non-flexible loads. Enforcing network-level constraints without full model knowledge highlights the potential of feedback-based pricing mechanisms to enhance robustness in DR mechanisms.

Further, the comparison among the different algorithms provides key insights into the trade-offs between model-based and feedback-based coordination. While the PO achieves near-optimal performance when perfect information is available, its practicality is limited in realistic scenarios where such information is hard to obtain. In contrast, the Stochastic InPO offers an alternative that can adapt to changing conditions and uncertainty in user behavior. 

\begin{figure}[tb]
\centering	
\includegraphics[width=0.9\columnwidth]{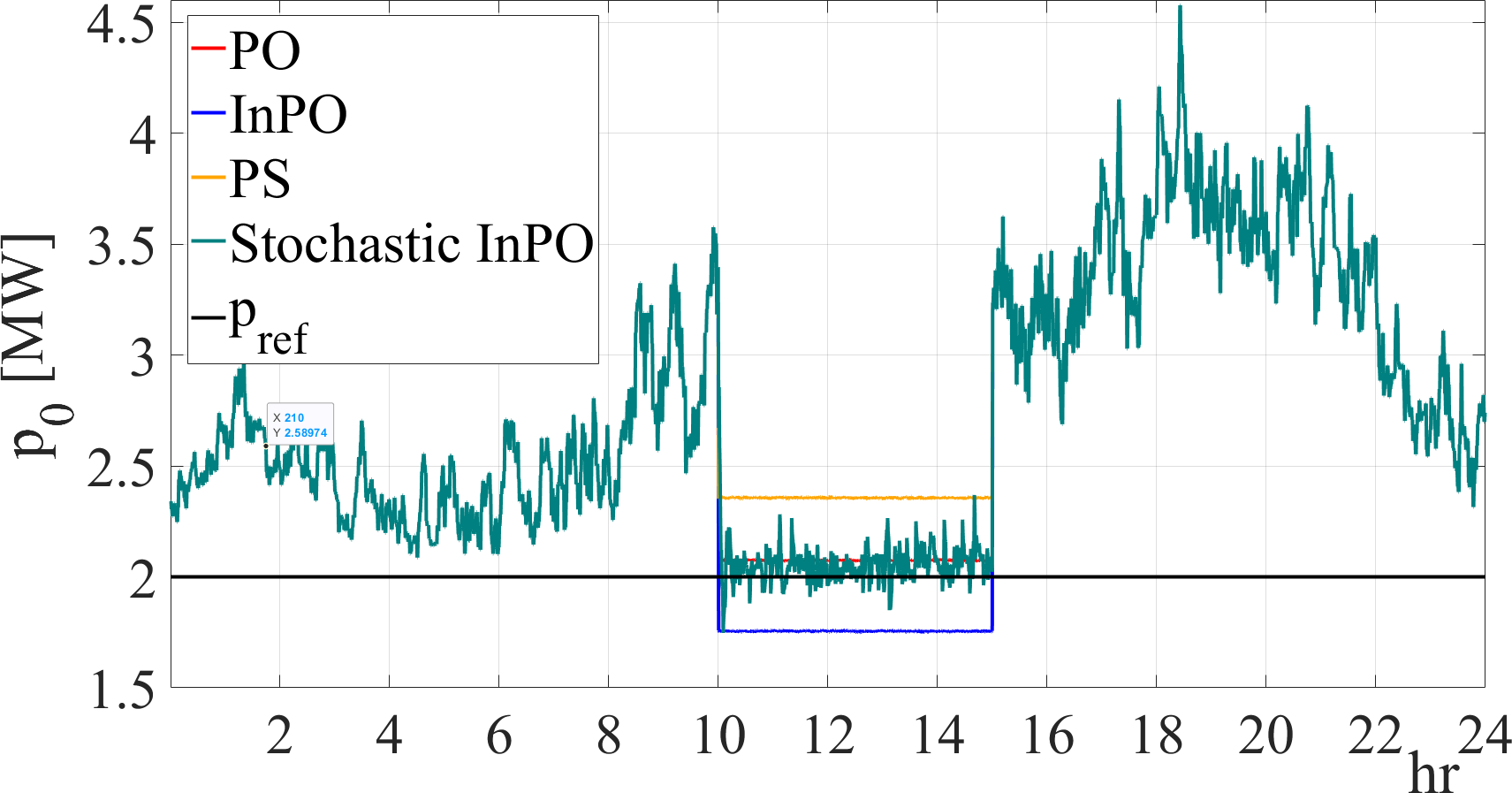}
\caption{Average power exchanged with the external network in the time-varying scenario.}
\label{fig:PO_MC_TVa}
\end{figure}

\begin{figure}[tb]
\centering	
\includegraphics[width=0.9\columnwidth]{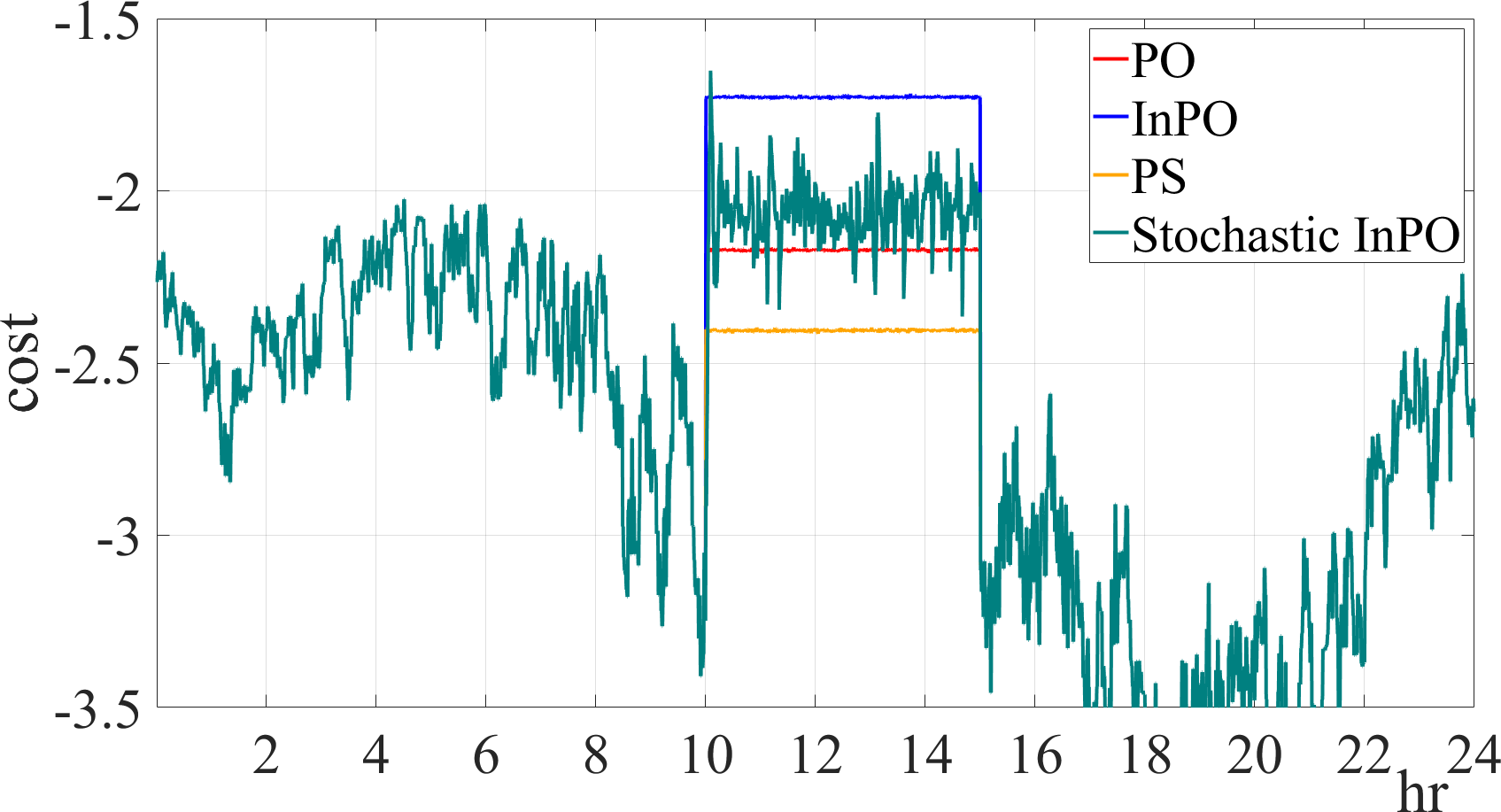}
\caption{Average operational cost in the time-varying scenario. }
\label{fig:Cost_MC_TV}
\vspace{-0.2cm}
\end{figure}

\section{Conclusion} \label{sec:conc}
This paper proposed a real-time price-based DR mechanism for regulating power exchanges in distribution networks under uncertain consumer behaviors. We proposed a solution that integrates real-time feedback from the grid to ensure convergence and constraint satisfaction even when user sensitivity estimates affected by errors.
We demonstrated  the effectiveness and robustness of the proposed method through  detailed case studies on the IEEE 37-bus feeder and comparison with model-based open-loop alternatives.
The results suggest that integrating feedback control and stochastic optimization represents a promising pathway for implementing responsive, grid-interactive energy systems.

\appendix
This Section provides the formal proofs of the results presented in Section~\ref{sec:control}.

\begin{IEEEproof}[Proof of Prop.~\ref{prop:static}]
The proof follows steps similar to~\cite{dall2016optimal}, but with key changes that pertains to the stochastic model and the errors in the gradient computation. Define the following map: 
\begin{align}
    \label{eq:approximatemap}
    \tilde{G}(\bz) := 
    \begin{bmatrix}
        \bg(\bx(t),\lambda(t)) \\
        -\xi(\bw(t)) - \eta
    \end{bmatrix}
\end{align}
which allows us to rewrite the stochastic iteration~\eqref{eq:PD_In_st} as: 
\begin{align}
\label{eq:PD_In_st_2}
\bz(t+1) = \text{proj}_{\mcZ} \left\{\bz(t) - \epsilon \tilde{G}(\bz(t)) \right\} \, .
\end{align}
In the following, we will find a bound on $e(t)$. Compute: 
\begin{align*}
    e(t+1) & = \Big\|\big[ \bz(t) - \epsilon \tilde{G}(\bz(t))\big]_{\mcZ} - \bz^\star \Big\| \\
    & = \Big \| \big[\bz(t) - \epsilon \tilde{G}(\bz(t))\big]_{\mcZ} - \big[\bz^\star - \epsilon G(\bz^\star)\big]_\mcZ \Big\| \\
    & \leq \|\bz(t) - \epsilon \tilde{G}(\bz(t)) - \bz^\star - \epsilon G(\bz^\star) \| \\
    & = \|\bz(t) - \epsilon \tilde{G}(\bz(t)) + \epsilon G(\bz(t)) \\
    & ~~~~~~~~~~~~~~~~~~~ - \epsilon G(\bz(t)) - \bz^\star - \epsilon G(\bz^\star)\|  \\
    & \leq \|\bz(t) - \epsilon G(\bz(t)) - \bz^\star - \epsilon G(\bz^\star)\| \\
    & ~~~~~~~~~~~~~~~~~~~ + \epsilon \| \tilde{G}(\bz(t)) - G(\bz(t))\| \\
    & \leq c(\epsilon) e_t + \epsilon \| \tilde{G}(\bz(t)) - G(\bz(t))\| 
\end{align*}
where $ c(\epsilon) := (1- 2 \epsilon \nu + \epsilon^2 L^2)^\frac{1}{2}$~\cite{koshal2011multiuser,dall2016optimal}.  Next, focus on $e_G(t) := \| \tilde{G}(\bz(t)) - G(\bz(t))\|$. Note that: 

\begin{small}
$$
e_G(t) = 
\left\|
\begin{bmatrix}
    2 (\hat{\bB} -\bB ) \bx(t) + (\pi - \pi_0)(\hat{\bbeta} - \bbeta) + \lambda(t) (\hat{\bbeta} - \bbeta) \\
        -\xi(\bw(t)) + \E_{\bw \sim \mcD(\bx(t))} \xi(\bw)  \, 
\end{bmatrix}
\right\| \, 
$$
\end{small}

Recall that for $\bz = [\bx^\top, \lambda ]^\top \in \mathbb{R}^{n+1}$, we have that $\|z\|_p = \big( \|x\|_p^p + |\lambda|^p \big)^{1/p}$, for any \(p\)-norm. Taking $p = 2$, we also have that $\|z\|_2 \;\leq\; \|x\|_2 + |\lambda|$. This is because, $\|z\|_2 = \sqrt{\|x\|_2^2 + |\lambda|^2} \leq \sqrt{\|x\|_2^2} + \sqrt{|\lambda|^2} = \|x\|_2 + |\lambda|$.  
Therefore, 
\begin{align*}
e_G(t) \leq & e_x(t) + e_\xi(\bx(t)) 
\end{align*}
where 
$$
e_x(t) := \|2 (\hat{\bB} -\bB ) \bx(t) + (\pi - \pi_0)(\hat{\bbeta} - \bbeta) + \lambda(t) (\hat{\bbeta} - \bbeta)\| 
$$
$$
e_\xi(\bx(t)) :=  \left|\xi(\bw(t)) - \E_{\bw \sim \mcD(\bx(t))} \xi(\bw) \right| \, .
$$
Note that: 
\begin{align*}
e_x(t)  & \leq \|\hat{\bbeta} - \bbeta\| (2 \|\bx(t)\| + |\pi - \pi_0| + |\lambda(t)|) \\
& \leq \|\hat{\bbeta} - \bbeta\| (2 X + |\pi - \pi_0| + \Lambda)
\end{align*}
where 
$$
X = \max_{\bx \in \mcX} \|\bx\| ,  \quad \Lambda = \max_{\lambda \in \mcY} |\lambda| \, .
$$
Note that $X$ and $\Lambda$ exists because $\mcX$ and $\mcY$ are compact. For brevity, define 
$$
B := 2 X + |\pi - \pi_0| + \Lambda \, .
$$
We can then bound $e(t+1)$ as: 
$$
 e(t+1) \leq c(\epsilon) e_t + \epsilon \|\hat{\bbeta} - \bbeta\| B + \epsilon e_\xi(\bx(t)) \, .
$$
Therefore, 
\begin{align}
    e(t+1) \leq & ~ c(\epsilon)^{t+1} e(0) + \epsilon \|\hat{\bbeta} - \bbeta\| B  \sum_{i = 0}^t c(\epsilon)^i  \nonumber \\
    &  + \epsilon \sum_{i = 0}^t c(\epsilon)^i \epsilon e_\xi(\bx(t-i)) \, . 
    \label{eq:proof_e1}
\end{align}

Then, using the Assumption~\ref{as:boundederror} in~\eqref{eq:proof_e1} and taking the expectation on both sides, we have that 
\begin{align}
    \E[e(&t+1)] \leq c(\epsilon)^{t+1} \E[e(0)] + \epsilon ( \|\hat{\bbeta} - \bbeta\| B  + \bar{e}_\xi )\sum_{i = 0}^t c(\epsilon)^i  \nonumber \\ \leq 
    & ~ c(\epsilon)^{t+1} \E[e(0)] + \epsilon ( \|\hat{\bbeta} - \bbeta\| B  + \bar{e}_\xi ) \frac{1 - c(\epsilon)^{t+1}}{1 - c(\epsilon)}
\end{align}
If $\epsilon < 2\nu/L^2$, then $c(\epsilon) < 1$ and the bound follows straightforwardly.  

\end{IEEEproof}

\begin{IEEEproof}[Proof of Prop.~\ref{prop:dynamic}]
The key steps are as follows: 
\begin{align*}
    e(t+1) & = \Big \| \big[\bz(t) - \epsilon \tilde{G}(\bz(t))\big]_{\mcZ} - \bz^\star(t+1) \Big\| \\
    & = \Big \|\big[\bz(t) - \epsilon \tilde{G}(\bz(t))\big]_{\mcZ} -  \bz^\star(t) + \bz^\star(t) - \bz^\star(t+1)\| \\
    & \leq \Big\| \big[\bz(t) - \epsilon \tilde{G}(\bz(t))\big]_{\mcZ} - \big[\bz^\star(t) - \epsilon \tilde{G}(\bz^\star(t))\big]_{\mcZ} \Big\| \\
    & \qquad + \|\bz^\star(t+1) - \bz^\star(t)\| \\
    & \leq \|\bz(t) - \epsilon \tilde{G}(\bz(t)) + \epsilon {G}(\bz(t)) - \epsilon {G}(\bz(t)) - \bz^\star(t) + \\
    & \qquad + \epsilon G(\bz^\star(t)) \| + \|\bz^\star(t+1) - \bz^\star(t)\| \\
    & \leq \|\bz(t) - \epsilon {G}(\bz(t)) - \bz^\star(t)  + \epsilon G(\bz^\star(t)) \| \\
    & \qquad + \epsilon \| \tilde{G}(\bz(t)) - {G}(\bz(t))\| + \|\bz^\star(t+1) - \bz^\star(t)\| \\
    & \leq c(\epsilon) e_t  + \|\bz^\star(t+1) - \bz^\star(t)\|  + \\
    & \qquad + \epsilon \| \tilde{G}(\bz(t)) - G(\bz(t))\|\, .
\end{align*}
By bounding $\|\bz^\star(t+1) - \bz^\star(t)\| $ with $\Delta$, using the bound for $\| \tilde{G}(\bz(t)) - G(\bz(t))\|$ derived in Proposition~\ref{prop:static}, and repeatedly applying the bound, the result follows. 

\end{IEEEproof}

\bibliographystyle{ieeetr}
\bibliography{Bib/myabrv,Bib/Bibliography}

\balance
% that's all folks
\end{document}